\documentclass[letterpaper, 10 pt, conference]{ieeeconf}  

\IEEEoverridecommandlockouts                              

\usepackage{amsmath} 
\usepackage{amssymb}  

\usepackage{amsfonts}

\usepackage{amsthm}
\usepackage{cite}

\theoremstyle{plain}
\newtheorem{assumption}{Assumption}
\newtheorem{remark}{Remark}
\newtheorem{proposition}{Proposition}

\theoremstyle{plain}
\newtheorem{theorem}{Theorem}
\theoremstyle{definition}
\newtheorem{definition}{Definition}
\theoremstyle{definition}

\theoremstyle{definition}
\newtheorem{problem}{Problem}

\title{\LARGE \bf
Design and Stability Analysis of a Shared Mobility Market
}

\author{Ioannis Vasileios Chremos, \textit{Student Member, IEEE}, and Andreas A. Malikopoulos, \textit{Senior Member, IEEE}
\thanks{This research was supported by the Sociotechnical Systems Center (SSC) at the University of Delaware.}
\thanks{The authors are with the Department of Mechanical Engineering, University of Delaware, Newark, DE 19716 USA (emails: \tt\small{ichremos@udel.edu}; \tt\small{andreas@udel.edu}.)}
}

\begin{document}

\maketitle
\thispagestyle{empty}
\pagestyle{empty}

\begin{abstract}

In recent years, we have witnessed a remarkable surge of usage in shared vehicles in our cities. Shared mobility offers a future of no congestion in busy city roads with increasing populations of travelers, passengers, and drivers. Given the behavioral decision-making of travelers and the shared vehicles' operators, however, the question is ``how can we ensure a socially-acceptable assignment between travelers and vehicles?" In other words, how can we design a shared mobility system that assigns each traveler to the ``right" vehicle? In this paper, we design a shared mobility market consisted of travelers and vehicles in a transportation network. We formulate a linear programming problem and derive the optimal assignment between travelers and vehicles. In addition, we provide the necessary and sufficient conditions for the stable traveler-vehicle profit allocation. Our goal is twofold: maximize the social welfare of all travelers with the optimal assignment while ensuring the feasibility and stability of the traveler-vehicle profit allocation while respecting the decision-making of both the travelers and the vehicles' operators.

\end{abstract}

\section{Introduction}

\subsection{Motivation}

Shared mobility can provide access to transportation on a custom basis without vehicle ownership \cite{becker2020}. Over the last few years, on-demand ride-sharing services available via our smartphone have proved to be an innovative and adaptive mobility strategy for a broad range of travelers, passengers, and drivers \cite{stocker2018}. Besides the apparent benefits to travelers (e.g., short-term and as-needed mobility access \cite{mourad2019}), shared mobility services have been shown to have a significant environmental and societal impact. For example, reduced vehicle use, ownership, and vehicle miles traveled \cite{shaheen2016}. However, it is the authors' belief that shared mobility can also provide a solution to the social impact of connected and automated vehicles (CAVs), which promise to be an incoming disruptive innovation with vast technological, commercial, and regulatory dimensions \cite{marletto2019}. Although it is clear that CAVs will transform the urban transportation networks and revolutionize mobility \cite{malikopoulos2020}, 
we expect CAVs to have social consequences. For example, CAVs may reshape urban mobility in the sense of altered tendency-to-travel and highly increase traffic demand \cite{lamnabhi2017}. To elaborate on this point, evident from similar technological revolutions (e.g., elevators \cite{bernard2014}), human social tendencies and society's perspective can change how a technological development is used and applied. Thus, one fundamental question that we need to ask ourselves is whether the deployment of CAVs in society will give rise to unexpected outcomes. For example, will the overall vehicle miles traveled increase to the point where we observe a decrease in traveler usage of public transit? Shared mobility can be a cost-effective and flexible mode of transportation alongside CAVs and provide mobility access to city travelers without increasing congestion, pollution, accidents, and energy consumption.

In this paper, we address shared mobility by designing a \emph{shared mobility market}. This market is consisted of a finite number of travelers and vehicles, and it is managed by a social planner. Our goal is to measure the ``benefit" received of both the travelers and the vehicles' operators, define the social welfare as a function of these benefits, and form a maximization problem with integer solutions subject to physically-related constraints. From a game-theoretic perspective, our proposed shared mobility market can be interpreted as an ``assignment game," in which indivisible goods are exchanged between two parties for money \cite{shapley1971}.

One may ask why use game theory to analyze such a problem. Emerging mobility systems (e.g., CAVs, shared mobility, electric vehicles) will be characterized by their socio-economic complexity: (i) improved productivity and energy efficiency, (ii) widespread accessibility, and (iii) drastic urban redesign and evolved urban culture. This characteristic can be naturally modeled and analyzed using notions from social choice theory and game theory. One of the main arguments in this paper is that the interaction between travelers and shared vehicles can be modeled as a market, in which we can find the socially-acceptable equilibrium by ensuring to take into account the most important factors that influence the travelers' decision-making. It is for this reason why the authors of this paper argue that game-inspired markets offer a complimentary analysis of decision-making in shared mobility systems.

Our current work is supported by previous work that investigated different yet related problems. Using methods and techniques from game theory, mechanism design, and insights from behavioral economics, in previous work, we modeled the human social interaction with vehicles as a social dilemma in a game-theoretic approach \cite{chremos2020_ITSC}. We investigated the \emph{social-mobility dilemma}, i.e., the binary decision-making of travelers between commuting with a vehicle or using public transportation. In \cite{chremos2020_CDC}, using mechanism design, we modeled the travel time of selfish travelers in a transportation network with vehicles as a \emph{social resource allocation problem}. We solved a social-welfare maximization problem by eliciting the private information (e.g., preferred travel time) of each traveler, and by considering a Nash-implementation approach, we showed that our informationally decentralized mechanism efficiently allocates travel time to all travelers that seek to commute in the network.

\subsection{Literature Review}


In recent years, it has been recognized in the literature that further research is required to identify and understand the potential impacts of emerging mobility \cite{sarkar2016,zhao2020}. Shared mobility in emerging mobility systems has been investigated and studied extensively the last decade. Factors that motivate active research in shared mobility systems are the significant energy savings \cite{zhao2020}, the limited importance of parking, and thus, opportunities for urban redesign with more space, and the increased demand for mobility access in developing countries \cite{shaheen2007}. Even though the promises of shared mobility have been realized with the implementation of various ride-sharing, or car-sharing, programs and initiatives, there are still open questions on how to design a shared mobility system that is socially-acceptable and profitable. Standard techniques of optimization and dynamics pricing have been used to control shared vehicle traffic and the non-strategic behavior of travelers and passengers \cite{pfrommer2014,waserhole2016}. These methods focus primarily on formulating and solving a dynamic or stochastic optimization problem with respect to variables that include preferred and expected departure, arrival, and in-vehicle travel. One can control the solution by designing pricing schemes that the travelers, or passengers, react in a predictable manner (travelers are assumed to be price-takers).

There have been different approaches reported in the literature to study shared mobility using ideas from game theory. In particular, game theory has been used to model and analyze non-cooperative, or cooperative, interactions of travelers who seek to accommodate their desired origin-destination commutes through ride-sharing. The authors in \cite{bistaffa2017} modeled a shared mobility system connected with a social network in which travelers could communicate and arrange one-time rides. Their focus was to minimize travel cost. Assignment games have been used to match sets of travelers with sets of capacitated routes in a transportation network \cite{rasulkhani2019}. In contrast to game-theoretic techniques, other efforts used a Vickrey–Clarke–Groves-inspired mechanism to design a first-mile, ride-sharing mobility system matching selfish travelers to vehicles \cite{bian2019}. The proposed mechanism was shown to be incentive compatible, individually rational, and price non-negative. In most cases, however, traveler, or passenger, behavior has not been well-understood. This is so, especially, in relevance to the impact that human drivers, passengers, and travelers might have on the traffic and energy efficiency of a mobility system. A very recent study on ``social dilemmas" attempted to remedy this lack of understanding on the social impact of shared mobility \cite{rapoport2019}. Informally, a social dilemma is any situation where there is a subtle yet unwanted discrepancy between individual and collective interest. The authors provided both a theoretical and an experimental study of how the strategic decision-making of travelers can impact the shared mobility's welfare, and thus, efficiency. A thorough review on ride-sharing can be found in \cite{furuhata2013} and the references therein.

Close in spirit to our approach is the long literature on stable matching problems. The seminal work in \cite{gale1962} was the first to study the ``marriage problem" for one-to-one matching, and the ``college admissions problem" for many-to-one matching. In less that a decade later, the authors in \cite{shapley1971} designed the so-called ``assignment game," which provides cooperative solutions for matching problems, where the players have transferable utilities. They achieved this by formulating and solving a linear optimization problem and by making sure that the solution is an equilibrium such that no one from the game's players prefers to deviate. A comprehensive outlook of assignments games can found in \cite{roth1990} and the references therein. More recently, an interesting development was reported in \cite{anshelevich2013}, where the provide bounds in the solutions of the assignment game that can result in stable mappings.


\subsection{Contribution of the Paper}

The main contribution of this paper is the design of a shared mobility market for the stable assignment of travelers to shared vehicles. By stable we mean that, considering the decision-making of both travelers and vehicles' operators, no other assignment is preferred. We formulate a linear optimization problem and we show that our shared mobility market can produce optimal assignments with feasible and stable traveler-vehicle profit allocation. For the latter, we also give the necessary and sufficient conditions when stability can be guaranteed.

\subsection{Organization of the Paper}

The paper is structured as follows. In Section \ref{SEC:formulation}, we present the mathematical formulation of our shared mobility market, which forms the basis for the rest of the paper. In Section \ref{SEC:properties}, we provide a feasibility and stability analysis of the shared mobility market and finally, in Section \ref{SEC:conclusion}, we draw conclusions and offer a discussion of future research.

\section{Mathematical Formulation}\label{SEC:formulation}

We consider a mobility system managed by a social planner whose objective is to assign $m \in \mathbb{N}$ vehicles to $n \in \mathbb{N}$ travelers, where $n \geq m$. We denote the set of travelers by $\mathcal{I} = \{0\} \cup \{1, 2, \dots, n\}$ and the set of vehicles by $\mathcal{J} = \{1, 2, \dots, m\}$. In $\mathcal{I}$, the index $0$ has no practical meaning other than helping us to assign any vehicles that have not been assigned to travelers. Travelers seek to travel in a transportation network represented by a directed multi-graph $\mathcal{G} = (\mathcal{V}, \mathcal{E})$, where each vertex in $\mathcal{V}$ represents a different city area, or neighborhood, and each edge $e \in \mathcal{E}$ represents a city road connection. In this network, an arbitrary traveler $i \in \mathcal{I}$ wants to travel from their current location $o_i \in \mathcal{V}$ to their self-chosen destination $d_i \in \mathcal{V}$. So, we say that traveler $i \in \mathcal{I}$ is associated with an origin-destination pair $(o_i, d_i)$. Similarly, each vehicle is associated with a route, i.e., a specific sequence of edges. Hence, the social planner aims to assign all the travelers that their $(o_i, d_i)$ can be satisfied by the vehicle's route.




\begin{definition}
    The \emph{traveler-service assignment} is a vector $\mathbf{a} = (a_{1 1}, \dots, a_{i j}, \dots, a_{n m}) = (a_{i j})_{{i \in \mathcal{I}}, j \in \mathcal{J}}$, where $a_{i j}$ is a binary variable of the form:
        \begin{equation}
            a_{i j} =
                \begin{cases}
                    1, \; & \text{if $i \in \mathcal{I}$ is assigned to $j \in \mathcal{J}$}, \\
                    0, \; & \text{otherwise}.
                \end{cases}
        \end{equation}
\end{definition}

A traveler $i$'s satisfaction is represented by a valuation function $v_i(a_{i j}) \in [\underline{v}_i, \bar{v}_i]$ when assigned to vehicle $j \in \mathcal{J}$, where $\underline{v}_i \in \mathbb{R}_{\geq 0}$ represents the lower bound of traveler $i$'s satisfaction, and $\bar{v}_i \in \mathbb{R}_{\geq 0}$ represents the upper bound of traveler $i$'s satisfaction. Intuitively, a traveler's satisfaction reflects the traveler's value of the service they expect to receive from a vehicle $j \in \mathcal{J}$. 

The satisfaction $v_i(\cdot)$ can be defined in terms of several factors (e.g., preferred and experienced number of co-travelers, in-vehicle travel time, or pickup time) that measure how satisfied the traveler can be with vehicle $j \in \mathcal{J}$. For example, a traveler can have a preferred travel time and their satisfaction can measure the monetary value of the difference between preferred and experienced travel time. So, the dis-utility caused by vehicle $j \in \mathcal{J}$ to traveler $i \in \mathcal{I}$ is given by $\phi_i(a_{i j}) \in \mathbb{R}_{\geq 0}$. We call $\phi_i(\cdot)$ the \emph{inconvenience cost} as it can measure the travel inconvenience caused to a traveler. Thus, we have
    \begin{equation}\label{EQN:valuation-function}
        v_i(a_{i j}) = \bar{v}_i - \phi_i(a_{i j}).
    \end{equation}
where $\bar{v}_i$ is the upper bound of traveler $i$'s satisfaction. Although our analysis will treat $v_i(a_{i j})$ in its most general form \eqref{EQN:valuation-function}, one can explicitly define $v_i(a_{i j})$ as follows, 
    \begin{equation}
        v_i(a_{i j}) =
            \begin{cases}
                \bar{v}_i, & \text{if } \phi_i = 0, \\
                \lambda_i \cdot \bar{v}_i, & \text{if } \phi_i = (1 - \lambda_i) \cdot \bar{v}_i, \\
                0, & \text{if } \phi_i = \bar{v}_i,
            \end{cases}
    \end{equation}
where $\lambda_i \in (0, 1)$ is a discount rate.

Next, the total utility of traveler $i \in \mathcal{I}$ is given by
    \begin{equation}\label{EQN:utility-function}
        u_i(a_{i j}) = v_i(a_{i j}) - t_i(a_{i j}),
    \end{equation}
where $t_i \in \mathbb{R}_{> 0}$ is the monetary payment that traveler $i \in \mathcal{I}$, e.g., a fare that traveler $i \in \mathcal{I}$ may make for the services of vehicle $j \in \mathcal{J}$.
Apparently, \eqref{EQN:utility-function} establishes a ``quasi-linear" relationship between a traveler's satisfaction and payment, both measured in monetary units.


\begin{definition}\label{DEFN:vehicle-capacity}
    For each vehicle $j \in \mathcal{J}$, the \emph{vehicle maximum capacity} $\varepsilon_j \in \mathbb{N}$ yields how many travelers can receive a ride from vehicle $j \in \mathcal{J}$.
\end{definition}


\begin{definition}
    The \emph{social welfare} of the shared mobility market is the collective summation of all travelers' utilities, i.e., $W(\mathbf{a}) = \sum_{i \in \mathcal{I}} u_i(a_{i j})$.
\end{definition}

As we will see in Subsection \ref{SUBSEC:problem-formulation}, our objective is to maximize the social welfare.

\begin{definition}\label{DEFN:operating-cost}
    Let the operating cost of vehicle $j \in \mathcal{J}$ denoted by $c_j \in \mathbb{R}_{> 0}$, which is shared (not necessarily equally) by each traveler $i \in \mathcal{I}$ assigned to vehicle $j \in \mathcal{J}$ as follows:
        \begin{equation}
            c_j = \sum_{i \in \mathcal{I} \setminus \{0\}} c_{i j}(a_{i j}),
        \end{equation}
    where $c_{i j}(a_{i j})$ is traveler $i$'s share of the operating cost of vehicle $j \in \mathcal{J}$.
\end{definition}

So far, we have described how the shared mobility market works to assign travelers to shared vehicles. Next, we explicitly define the ``end" of our market in terms of monetary payments and net profits for both the travelers and the vehicles.

\begin{definition}\label{DEFN:traveler-vehicle-profits}
    At the end of travel, each traveler is asked to make a payment $t_i(a_{i j})$ for the service of vehicle $j \in \mathcal{J}$ (e.g., a share-mobility fare). The monetary net profit $\rho_{i j}(a_{i j})$ of vehicle $j \in \mathcal{J}$ from traveler $i \in \mathcal{I}$ is given by
        \begin{equation}\label{EQN:vehicle-net-profit}
            \rho_{i j}(a_{i j}) = t_i(a_{i j}) - c_{i j}(a_{i j}).
        \end{equation}
    On the other hand, the monetary net profit of traveler $i \in \mathcal{I}$ is
        \begin{equation}\label{EQN:traveler-net-profit}
            \pi_{i j}(a_{i j}) = v_i(a_{i j}) - t_i(a_{i j}) - \underline{v}_i,
        \end{equation}
    We call $\left( \pi_{i j}(a_{i j}), \rho_{i j}(a_{i j}) \right)_{i \in \mathcal{I}, j \in \mathcal{J}}$ the \emph{traveler-vehicle profit allocation}.
\end{definition}

\begin{remark}
    Naturally, \eqref{EQN:vehicle-net-profit} gives the net profit of a vehicle $j \in \mathcal{J}$ generated by one traveler $i \in \mathcal{I}$ as the difference between the monetary payment $t_i$ (e.g., fare) made by the traveler, and the traveler's share of the operating cost, $c_{i j}$. In a similar line of arguments, in \eqref{EQN:traveler-net-profit} the net profit of traveler $i \in \mathcal{I}$ is the difference between what they are willing to pay, $v_i$, what they actually pay, $t_i$, and the minimum accepted value that they expect to get from vehicle $j \in \mathcal{J}$.
\end{remark}


Next, following a similar notion from \cite{sotomayor1992}, we define when the traveler-vehicle profit allocation $\left( \pi_{i j}(a_{i j}), \rho_{i j}(a_{i j}) \right)$ for each traveler $i \in \mathcal{I}$ and each vehicle $j \in \mathcal{J}$ is feasible.

\begin{definition}\label{DEFN:feasible-outcomes}
    Let $\widehat{\mathcal{J}} \subseteq \mathcal{J}$ denote the set of all vehicles that are actually assigned to travelers. We say $\left( \pi_{i j}(a_{i j}), \rho_{i j}(a_{i j}) \right)_{i \in \mathcal{I}, j \in \mathcal{J}}$ is feasible if (1) for all vehicles $j \in \widehat{\mathcal{J}}$, both the traveler's and vehicle's net profit are nonnegative, i.e., $\pi_{i j}(a_{i j}), \rho_{i j}(a_{i j}) \geq 0$; (2) the net profit of any traveler $i \in \mathcal{I}$ assigned to any vehicle $j \in \mathcal{J}$ and its net profit is equal to the total utility of traveler $i \in \mathcal{I}$ minus the operating cost of vehicle $j \in \mathcal{J}$, i.e.,
        \begin{equation}\label{EQN:feasibility-equation}
            \pi_{i j}(a_{i j}) + \rho_{i j}(a_{i j}) = u_i(a_{i j}) - c_{i j}(a_{i j}); 
        \end{equation}
    (3) for all unassigned vehicles $j \in \mathcal{J} \setminus \widehat{\mathcal{J}}$, $\rho_{i j}(a_{i j}) = 0$; and (4) for any traveler $i \in \mathcal{I}$ left unassigned, $\pi_{i j}(a_{i j}) = 0$.
\end{definition}

\begin{definition}\label{DEFN:stability}
    A feasible traveler-vehicle profit allocation $\left( \pi_{i j}(a_{i j}), \rho_{i j}(a_{i j}) \right)_{i \in \mathcal{I}, j \in \mathcal{J}}$ is stable if for all $i \in \mathcal{I}$
        \begin{equation}\label{EQN:defn-2}
            u_i(a_{i j}) - c_{i j}(a_{i j}) \geq u_i(a_{i j} ') - c_{i j}(a_{i j} '),
        \end{equation}
        for any assignment $a_{i j} '$.
\end{definition}

In other words, Definition \ref{DEFN:stability} implies that for any traveler $i$ and any vehicle $j$ that are not assigned together, if $u_i(a_{i j}) - c_{i j}(a_{i j}) < u_i(a_{i j} ') - c_{i j}(a_{i j} ')$, then neither traveler $i$ or vehicle $j$ would be satisfied with that assignment. If we can eliminate those cases, then the traveler-vehicle profit allocation is socially-acceptable and no traveler, or vehicle, will seek to deviate.


\subsection{Assumptions}

In our modeling framework to design the shared mobility market we impose the following assumptions.

\begin{assumption}\label{ASM:no-alternatives}
    All travelers participate in the market since sharing a vehicle is the only commute option.
\end{assumption}

We impose Assumption \ref{ASM:no-alternatives} in our modeling framework since the focus is on identifying the best assignment between travelers and shared vehicles. By including alternative commute options, we would just add complexity in our analysis without any compelling reason. However, in future work, we plan to relax this assumption, and allow travelers to have multiple commute options using different modes of transportation to reach their destination in the network.

\begin{assumption}\label{ASM:all-monetary}
    The travel satisfaction or costs of any traveler's utility is represented in monetary units. Also, we have $u_0(a_{0 j}) = c_j$ for any vehicle $j \in \mathcal{J}$.
\end{assumption}

Although Assumption \ref{ASM:all-monetary} allows us to simplify the mathematical modeling, it is also natural in a realistic market of shared mobility to assume that all valuations and transactions between travelers and vehicles are done using money. We assume $u_0(a_{0 j}) = c_j$ for any vehicle $j \in \mathcal{J}$ to ensure the vehicle's operating cost is covered by an assignment with a traveler.


\begin{assumption}\label{ASM:total-operating-cost}
    The total operating cost of all vehicles $\sum_{j \in \mathcal{J}} c_j$ remains fixed.
\end{assumption}

Assumption \ref{ASM:total-operating-cost} implies that the operating cost of all vehicles can not really altered in a long run while accommodating the travelers' desired origin-destination requests. In other words, the travel-service assignments can not really alter the total operating cost of all vehicles.

\subsection{Problem Formulation}\label{SUBSEC:problem-formulation}

\begin{problem}\label{PROB:centralized-problem}
    The optimization problem formulation of the shared mobility market is
        \begin{gather}
            \max_{a_{i j}} W(\mathbf{a}) = \max_{a_{i j}} \sum_{i \in \mathcal{I}} u_i(a_{i j}), \label{EQN:prob1-objective} \\
            \text{subject to:} \notag \\
            \sum_{j \in \mathcal{J}} a_{i j} \leq 1, \quad \forall i \in \mathcal{I}, \label{CSTR:prob1-1} \\
            \sum_{i \in \mathcal{I}} a_{i j} \leq \varepsilon_j, \quad \forall j \in \mathcal{J}, \label{CSTR:prob1-2} 
        \end{gather}
    where \eqref{CSTR:prob1-1} ensures that each traveler $i \in \mathcal{I}$ is assigned to only one vehicle $j \in \mathcal{J}$, and \eqref{CSTR:prob1-2} ensures that the vehicle maximum capacity is not exceeded while the vehicle shared by travelers.
\end{problem}


\begin{remark}
    We note that the solution of Problem \ref{PROB:centralized-problem} will always assign a vehicle that can satisfy the origins and destinations of all the travelers that are assigned to it.
\end{remark}



\section{Main Results}\label{SEC:properties}


\begin{theorem}\label{PROP:equivalence}
    Let $\mathbf{a} ^ *$ denote the optimal assignment of Problem \ref{PROB:centralized-problem}. Then, the objective function \eqref{EQN:prob1-objective} evaluated at $\mathbf{a} ^ *$ is mathematically equivalent to the classic maximization of the social welfare at $\mathbf{a} ^ *$ with utility function defined as
        \begin{equation}
            u_i(a_{i j}) = v_i(a_{i j}) - t_i(a_{i j}) - c_{i j}(a_{i j}),
        \end{equation}
    where $v_i(a_{i j})$ is the satisfaction of traveler $i \in \mathcal{I}$, $t_i(a_{i j})$ is the monetary payment made by traveler $i \in \mathcal{I}$ for using vehicle $j \in \mathcal{J}$, and $c_{i j}(a_{i j})$ is the operating cost of vehicle $j \in \mathcal{J}$ assigned to traveler $i \in \mathcal{I}$.
\end{theorem}

\begin{proof}
    By Assumption \ref{ASM:all-monetary}, we can write the objective function \eqref{EQN:prob1-objective} of Problem \ref{PROB:centralized-problem} as follows
        \begin{equation}\label{EQN:prop1-first}
            \max_{a_{i j}} \sum_{i \in \mathcal{I}} u_i(a_{i j}) = \max_{a_{i j}} \sum_{i \in \mathcal{I} \setminus \{0\}} u_i(a_{i j}) + \sum_{j \in \mathcal{J}} c_{0 j}(a_{0 j}),
        \end{equation}
    where the term $\sum_{j \in \mathcal{J}} c_{0 j}(a_{0 j})$ represents the total operating cost of all the vehicles that are unassigned to travelers, and can be written as
        \begin{equation}\label{EQN:prop1-second}
            \sum_{j \in \mathcal{J}} c_{0 j}(a_{0 j}) = \sum_{j \in \mathcal{J}} c_j - \sum_{j \in \widehat{\mathcal{J}}} c_j.
        \end{equation}
    Substituting \eqref{EQN:prop1-second} into \eqref{EQN:prop1-first} yields
        \begin{multline}\label{EQN:prop1-third}
            \max_{a_{i j}} \sum_{i \in \mathcal{I}} u_i(a_{i j}) = \max_{a_{i j}} \sum_{i \in \mathcal{I} \setminus \{0\}} u_i(a_{i j}) \\
            + \sum_{j \in \mathcal{J}} c_j - \sum_{j \in \widehat{\mathcal{J}}} c_j.
        \end{multline}
    Since $\sum_{j \in \mathcal{J}} c_j$ is constant by Assumption \ref{ASM:total-operating-cost}, it can be neglected from the maximization problem. Hence, by optimality, we have $\sum_{j \in \widehat{\mathcal{J}}} c_j = \sum_{j \in \mathcal{J}} \sum_{i \in \mathcal{I} \setminus \{0\}} c_{i j}(a_{i j})$, and since the series is finite, \eqref{EQN:prop1-third} becomes
        \begin{multline}\label{EQN:prop1-last}
            \max_{a_{i j}} \sum_{i \in \mathcal{I} \setminus \{0\}} u_i(a_{i j}) - \sum_{i \in \mathcal{I} \setminus \{0\}} \sum_{j \in \mathcal{J}} c_{i j}(a_{i j}) = \\
            \max_{a_{i j}} \sum_{i \in \mathcal{I} \setminus \{0\}} \left( v_i(a_{i j}) - t_i(a_{i j}) - c_{i j}(a_{i j}) \right),
        \end{multline}
    where in the last equation we have used the fact that $\sum_{i \in \mathcal{I} \setminus \{0\}} u_i(a_{i j}) = \sum_{i \in \mathcal{I} \setminus \{0\}} \sum_{j \in \mathcal{J}} u_i(a_{i j})$, and the result immediately follows.
\end{proof}

\begin{proposition}
    Let the traveler-vehicle profit allocation $\left( \pi_{i j}(a_{i j}), \rho_{i j}(a_{i j}) \right)_{i \in \mathcal{I}, j \in \mathcal{J}}$ under the traveler-vehicle assignment $\mathbf{a}$ of Problem \ref{PROB:centralized-problem} form a space, denoted by $\mathcal{S}$. Then, $\mathcal{S}$ is convex.
\end{proposition}

\begin{proof}
    It is straightforward to see that the space of stable solutions $\mathcal{S}$ is defined by a set of linear constraints. Therefore, the space of stable solutions $\mathcal{S}$ is convex. 
\end{proof}


\begin{theorem}[Stability]
    If $\left( \pi_{i j}(a_{i j}), \rho_{i j}(a_{i j}) \right)_{i \in \mathcal{I}, j \in \mathcal{J}}$ is stable, then $\mathbf{a}$ is an optimal assignment of Problem \ref{PROB:centralized-problem}.
\end{theorem}

\begin{proof}
    Let $\mathbf{a}$ and $\mathbf{a} '$ denote two different traveler-vehicle assignments of Problem \ref{PROB:centralized-problem}. It is sufficient to consider the case where $\left( \pi_{i j}(a_{i j}), \rho_{i j}(a_{i j}) \right)_{i \in \mathcal{I}, j \in \mathcal{J}}$ is stable under $\mathbf{a}$ and only feasible under $\mathbf{a} '$. Then, we want to show that $\mathbf{a} '$ is not optimal. So, by Definition \ref{DEFN:stability}, we have
        \begin{align}\label{EQN:prop3-first}
            \pi_{i j}(a_{i j}) + \rho_{i j}(a_{i j}) & = u_i(a_{i j}) - c_{i j}(a_{i j}) \notag \\
            & \geq u_i(a_{i j} ') - c_{i j}(a_{i j} ').
        \end{align}
    We take the summation over $i \in \mathcal{I} \setminus \{0\}$ and $j \in \widehat{\mathcal{J}}$ of \eqref{EQN:prop3-first} as follows
        \begin{multline}\label{EQN:prop3-second}
            \sum_{i \in \mathcal{I} \setminus \{0\}} \sum_{j \in \widehat{\mathcal{J}}} \left( \pi_{i j}(a_{i j}) + \rho_{i j}(a_{i j}) \right) \geq \\
            \sum_{i \in \mathcal{I} \setminus \{0\}} \sum_{j \in \widehat{\mathcal{J}}} \left( u_i(a_{i j} ') - c_{i j}(a_{i j} ') \right).
        \end{multline}
    So, the RHS of \eqref{EQN:prop3-second} becomes
        \begin{multline}\label{EQN:prop3-third}
            \sum_{i \in \mathcal{I} \setminus \{0\}} \sum_{j \in \widehat{\mathcal{J}}} \left( u_i(a_{i j} ') - c_{i j}(a_{i j} ') \right) = \\
            \sum_{i \in \mathcal{I} \setminus \{0\}} \sum_{j \in \mathcal{J}} \left( u_i(a_{i j} ') - (1 - a_{0 j} ') \cdot c_{i j}(a_{i j} ') \right).
        \end{multline}
    By using conditions (ii) and (iii) in Definition \ref{DEFN:feasible-outcomes}, the LHS of \eqref{EQN:prop3-second} becomes
        \begin{multline}\label{EQN:prop3-fourth}
            \sum_{i \in \mathcal{I} \setminus \{0\}} \sum_{j \in \widehat{\mathcal{J}}} \left( \pi_{i j}(a_{i j}) + \rho_{i j}(a_{i j}) \right) = \\
            \sum_{i \in \mathcal{I} \setminus \{0\}} \sum_{j \in \mathcal{J}} \left( u_i(a_{i j}) - (1 - a_{0 j}) \cdot c_{i j}(a_{i j}) \right).
        \end{multline}
    Thus, substituting \eqref{EQN:prop3-third} and \eqref{EQN:prop3-fourth} into \eqref{EQN:prop3-second} yields
        \begin{multline}
            \sum_{i \in \mathcal{I} \setminus \{0\}} \sum_{j \in \mathcal{J}} \left( u_i(a_{i j}) - (1 - a_{0 j}) \cdot c_{i j}(a_{i j}) \right) \geq \\
            \sum_{i \in \mathcal{I} \setminus \{0\}} \sum_{j \in \mathcal{J}} \left( u_i(a_{i j} ') - (1 - a_{0 j} ') \cdot c_{i j}(a_{i j} ') \right),
        \end{multline}
    which simplifies to, for any assignment $a_{i j} '$,
        \begin{equation}
            \sum_{i \in \mathcal{I}} u_i(a_{i j}) \geq \sum_{i \in \mathcal{I}} u_i(a_{i j} '),
        \end{equation}
    since the summation over the $j \in \mathcal{J}$ is redundant. Hence, the social welfare under assignment $\mathbf{a}$ is greater or equal than the social welfare under $\mathbf{a} '$. Therefore, we conclude that if $\left( \pi_{i j}(a_{i j}), \rho_{i j}(a_{i j}) \right)_{i \in \mathcal{I}, j \in \mathcal{J}}$ is stable, then the assignment $\mathbf{a}$ is necessarily optimal.
\end{proof}


\begin{theorem}
    If there are two optimal assignments of Problem \ref{PROB:centralized-problem}, denoted by $\mathbf{a}$ and $\tilde{\mathbf{a}}$, respectively, then the resulted traveler-vehicle profit allocation $\left( \pi_{i j}(a_{i j}), \rho_{i j}(a_{i j}) \right)_{i \in \mathcal{I}, j \in \mathcal{J}}$ is feasible and stable under both assignments.
\end{theorem}

\begin{proof}
    Let $\mathbf{a}$ and $\tilde{\mathbf{a}}$ denote two optimal assignment. What we have to show is that if $\left( \pi_{i j}(a_{i j}), \rho_{i j}(a_{i j}) \right)_{i \in \mathcal{I}, j \in \mathcal{J}}$ is stable under assignment $\mathbf{a}$ and feasible under $\tilde{\mathbf{a}}$, then it is also stable under $\tilde{\mathbf{a}}$. We follow the same arguments up until \eqref{EQN:prop3-second} to get
        \begin{multline}\label{EQN:prop4-first}
            \sum_{i \in \mathcal{I} \setminus \{0\}} \sum_{j \in \widehat{\mathcal{J}}} \left( \pi_{i j}(a_{i j}) + \rho_{i j}(a_{i j}) \right) \geq \\
            \sum_{i \in \mathcal{I} \setminus \{0\}} \sum_{j \in \widehat{\mathcal{J}}} \left( u_i(a_{i j} ') - c_{i j}(a_{i j} ') \right).
        \end{multline}
    Hence, we observe that if $\tilde{\mathbf{a}}$ is an optimal assignment, then by Definition \ref{DEFN:stability}, \eqref{EQN:prop4-first} will hold at equality. Thus, the feasibility equation \eqref{EQN:feasibility-equation} is satisfied. Therefore, under the optimal assignment $\tilde{\mathbf{a}}$ we conclude that $\left( \pi_{i j}(a_{i j}), \rho_{i j}(a_{i j}) \right)_{i \in \mathcal{I}, j \in \mathcal{J}}$ is stable.
\end{proof}

\begin{proposition}
    If there are two arbitrary travelers with the same needs that are assigned to different vehicles, then there is no difference in their utility.
\end{proposition}

\begin{proof}
    Suppose there are two travelers $i, i ' \in \mathcal{I}$ with the same needs and two vehicles $j, j ' \in \mathcal{J}$. We want to show that in our market both travelers will receive the same utility even under different assignments. So, we assume that there are two assignments $\mathbf{a}$ and $\mathbf{a} '$, where in $\mathbf{a}$ traveler $i \in \mathcal{I}$ is assigned to vehicle $j \in \mathcal{J}$ while in $\mathbf{a} '$ traveler $i \in \mathcal{I}$ is assigned to vehicle $j '$. Similarly, for traveler $i '$. For an optimal $\mathbf{a}$, the stability conditions of $\left( \pi_{i j}(a_{i j}), \rho_{i j}(a_{i j}) \right)_{i \in \mathcal{I}, j \in \mathcal{J}}$ are
        \begin{multline}\label{EQN:prop5-first}
            \pi_{i j}(a_{i j}) + \sum_{\ell \in \mathcal{I} \setminus \{i\}} \left( \pi_{\ell j}(a_{\ell j}) + \rho_{i j}(a_{i j}) \right) = \\
            u_{i}(a_{i j}) + \sum_{\ell \in \mathcal{I} \setminus \{i\}} \left( u_i(a_{\ell j}) - c_{\ell j}(a_{\ell j}) \right),
        \end{multline}
        
        \begin{multline}
            \pi_{i j}(a_{i j}) + \sum_{\ell \in \mathcal{I} \setminus \{i '\}} \left( \pi_{\ell j}(a_{\ell j}) + \rho_{i j '}(a_{i j '}) \right) \geq \\
            u_{i}(a_{i j '}) + \sum_{\ell \in \mathcal{I} \setminus \{i '\}} \left( u_i(a_{\ell j '}) - c_{\ell j '}(a_{\ell j '}) \right),
        \end{multline}
    Similarly, for traveler $i '$, we have the following:
        \begin{multline}
            \pi_{i ' j}(a_{i ' j}) + \sum_{\ell \in \mathcal{I} \setminus \{i '\}} \left( \pi_{\ell j}(a_{\ell j}) + \rho_{i j '}(a_{i j '}) \right) \geq \\
            u_{i}(a_{i ' j '}) + \sum_{\ell \in \mathcal{I} \setminus \{i '\}} \left( u_i(a_{\ell j '}) - c_{\ell j '}(a_{\ell j '}) \right),
        \end{multline}
        
        \begin{multline}
            \pi_{i ' j}(a_{i ' j}) + \sum_{\ell \in \mathcal{I} \setminus \{i\}} \left( \pi_{\ell j}(a_{\ell j}) + \rho_{i j}(a_{i j}) \right) = \\
            u_{i}(a_{i ' j}) + \sum_{\ell \in \mathcal{I} \setminus \{i\}} \left( u_i(a_{\ell j}) - c_{\ell j}(a_{\ell j}) \right),
        \end{multline}
    In a similar way, we can argue that since $\mathbf{a} '$ is optimal, the stability conditions of $\left( \pi_{i j}(a_{i j}), \rho_{i j}(a_{i j}) \right)_{i \in \mathcal{I}, j \in \mathcal{J}}$ are
        \begin{multline}
            \pi_{i j}(a_{i j}) + \sum_{\ell \in \mathcal{I} \setminus \{i\}} \left( \pi_{\ell j}(a_{\ell j}) + \rho_{i j '}(a_{i j '}) \right) = \\
            u_{i}(a_{i j '}) + \sum_{\ell \in \mathcal{I} \setminus \{i\}} \left( u_i(a_{\ell j '}) - c_{\ell j '}(a_{\ell j '}) \right),
        \end{multline}
        
        \begin{multline}
            \pi_{i j}(a_{i j}) + \sum_{\ell \in \mathcal{I} \setminus \{i '\}} \left( \pi_{\ell j}(a_{\ell j}) + \rho_{i j}(a_{i j}) \right) \geq \\
            u_{i}(a_{i j}) + \sum_{\ell \in \mathcal{I} \setminus \{i '\}} \left( u_i(a_{\ell j}) - c_{\ell j}(a_{\ell j}) \right),
        \end{multline}
    Similarly, for traveler $i '$, we have the following:
        \begin{multline}
            \pi_{i ' j}(a_{i ' j}) + \sum_{\ell \in \mathcal{I} \setminus \{i '\}} \left( \pi_{\ell j}(a_{\ell j}) + \rho_{i j}(a_{i j}) \right) = \\
            u_{i}(a_{i ' j}) + \sum_{\ell \in \mathcal{I} \setminus \{i '\}} \left( u_i(a_{\ell j}) - c_{\ell j}(a_{\ell j}) \right),
        \end{multline}

        \begin{multline}
            \pi_{i ' j}(a_{i ' j}) + \sum_{\ell \in \mathcal{I} \setminus \{i\}} \left( \pi_{\ell j}(a_{\ell j}) + \rho_{i j '}(a_{i j '}) \right) \geq \\
            u_{i}(a_{i ' j '}) + \sum_{\ell \in \mathcal{I} \setminus \{i\}} \left( u_i(a_{\ell j '}) - c_{\ell j '}(a_{\ell j '}) \right). \label{EQN:prop5-last}
        \end{multline}
    Recall that both travelers $i, i ' \in \mathcal{I}$ have the same needs. Thus, $u_i(a_{i j}) = u_{i '}(a_{i ' j})$ and $u_i(a_{i j '}) = u_{i '}(a_{i ' j '})$. Therefore, from \eqref{EQN:prop5-first} - \eqref{EQN:prop5-last}, it follows that $\pi_{i j}(a_{i j}) = \pi_{i ' j}(a_{i ' j})$.
\end{proof}

\section{Conclusions}\label{SEC:conclusion}

In this paper, we provided an answer to the question of how one can ensure a socially-acceptable assignment between travelers and the shared vehicles' operators. We focused on the behavioral decision-making of both the travelers and the vehicles' operators and designed a shared mobility market consisted of travelers and vehicles in a transportation network. We formulated a linear programming problem and derived necessary and sufficient conditions for its solution to be an assignment between travelers and vehicles that cannot be improved any further. Consequently, we showed that our optimal assignment maximizes the social welfare of all travelers, and ensures the feasibility and stability of the traveler-vehicle profit allocation while respecting the decision-making of both the travelers and the vehicles' operators.

Ongoing work includes the design of a mobility market with strategic decision-making and by relaxing the assumption that travelers have no other alternative modes of transportation. An interesting research direction would involve to extend and enhance the traveler-behavioral model, motivated by a social-mobility survey. The objective with such a survey would be to observe any correlations between behavioral tendencies or attitudes of travelers and how they use shared vehicles (e.g., Uber, Lyft, taxicabs). Some unanswered questions that we plan to address include ``In emerging mobility systems, how likely are people to share CAVs?" and ``Will CAVs play any role or have a significant impact on travelers' tendencies and behavior?"

\addtolength{\textheight}{-12cm}   








\bibliographystyle{IEEEtran}
\bibliography{references}

\end{document}